\documentclass{ws-mpla}
\usepackage{url}
\usepackage{hyphenat}
\usepackage[super]{cite}
\usepackage{graphicx}

\newcommand{\om}{{\vec{\omega}}}
\newcommand{\R}{{\vec{\rm R}}}
\newcommand{\xv}{{\vec{\rm x}}}
\newcommand{\ka}{{\vec{\rm k}}}

\newcommand{\q}{{\vec{\rm q}}}

\newcommand{\vs}{{\vec{\rm v}}}
\newcommand{\n}{{\vec{\rm n}}}

\newcommand{\jv}{\vec{\rm J}}

\newcommand%
{\MVV}%
[1]%
{{\langle\!\langle #1\rangle\!\rangle}}%
\newcommand%
{\MMV}%
[1]%
{{\left\langle\!\!\!\left\langle #1\right\rangle\!\!\!\right\rangle}}%
\newcommand{\bea}{\begin{eqnarray}}
\newcommand{\eea}{\end{eqnarray}}
\def\bldmth#1{%
\mathchoice
{{\hbox{\boldmath$\displaystyle#1$\unboldmath}}}%
{{\hbox{\boldmath$\textstyle#1$\unboldmath}}}%
{{\hbox{\boldmath$\scriptstyle#1$\unboldmath}}}%
{{\hbox{\boldmath$\scriptscriptstyle#1$\unboldmath}}}%
}
\def\vec#1{\bldmth{#1}}
\begin{document}

\markboth{S.~E.~Korenblit and A.~V.~Sinitskaya}
{On asymptotic power corrections to differential fluxes and   
generalization of optical theorem for potential scattering}

\catchline{}{}{}{}{}

\title{On asymptotic power corrections to differential fluxes and 
generalization of optical theorem for potential scattering}

\author{\footnotesize S.~E.~Korenblit${}^*$ and A.~V.~Sinitskaya${}^\dagger$}

\address{
${}^*$ Department of Physics, Irkutsk State University, \\
20 Gagarin blvd, Irkutsk 664003, Russia \\
${}^\dagger$ Irkutsk National Research Technical University, \\
83 Lermontov str., Irkutsk 664074, Russia \\
${}^*$korenb@ic.isu.ru  $ \quad {}^\dagger$lasalvia@mail.ru
}

\maketitle

\begin{abstract}
In a wide class of potentials the exact asymptotic dependence on finite distance $R$ 
from scattering center is established for outgoing differential flux. 
It is shown how this dependence is eliminated by integration over solid angle 
for total flux, unitarity relation, and optical theorem. 
Thus, their applicability domain extends naturally to the finite $R$.

\end{abstract}

\ccode{PACS Nos.:  03.65.-w, 03.65.Bz, 03.65.Nk, 11.80.Et,  11.80.Fv}

\section{Introduction}
According to the common rules, \cite{GW,S,T,Ng,AR} the differential cross-section 
$d\sigma$ for the scattering on hermitian scalar spherically symmetric potential $U(R)$ is 
uniquely defined by on-shell scattering amplitudes $f^\pm(\q;\ka)$. These amplitudes 
are defined as coefficients at outgoing or incoming spherical waves being the first order 
terms of asymptotic expansion of the scattering wave functions $\Psi^\pm_\ka(\R)$ for  
$R=|\R|\to\infty$, $\R=R\n$, $\q=k\n$, $\ka=k\om$, $\n^2=\om^2=1$: 
\begin{eqnarray}
&&\!\!\!\!\!\!\!\!\!\!\!\!\!\!\!\!\!\!\!\!\!\!\!\! 
\Psi^\pm_\ka(\R)
\underset{R\to\infty}{\longrightarrow} 
e^{i(\ka\cdot\R)} + f^\pm(\q;\ka)\frac{e^{\pm ikR}}R+O(R^{-2}),
\label{1_} \\
&&\!\!\!\!\!\!\!\!\!\!\!\!\!\!\!\!\!\!\!\!\!\!\!\! 
\Psi^-_\ka(\R)=\left(\Psi^+_{-\ka}(\R)\right)^*, \quad 
f^-(\q;\ka)=\left(f^+(\q;-\ka)\right)^*,
\label{2_} \\
&&\!\!\!\!\!\!\!\!\!\!\!\!\!\!\!\!\!\!\!\!\!\!\!\! 
d\sigma =\left|f^+(\q;\ka)\right|^2 d\Omega(\n),\;\mbox{ where: }\;  
\sigma =\int \left|f^+(k\n;k\om)\right|^2 d\Omega(\n),
\label{3_}
\end{eqnarray}
is the respective total (elastic) cross-section, which also does not depend on $R$. 
Of course the terms of order $O(R^{-2})$ in Eq. (\ref{1_}) are unimportant 
\cite{GW,S,T,Ng,AR} for both definitions (\ref{3_}). However, $R$ is finite for real 
experiments, and the recent investigations \cite{N_shk,NN_shk,k_t3,anom} of (anti-) 
neutrino processes at short distances from the source reveal a possible violation of 
inverse-square law for event rate corresponding \cite{N_shk,NN_shk} to (\ref{1_}) and  
(\ref{3_}). 
Since the macroscopic parameter of distance $R$ has very peculiar meaning when it is 
considered in the framework of quantum field theory \cite{N_shk,NN_shk,k_t3,Grimus}, 
it seems natural and convenient to elucidate this problem at first for nonrelativistic 
quantum-mechanical scattering. 

In the following sections the closed formula and recurrent relation for coefficients of 
asymptotic expansion of wave function $\Psi^\pm_\ka(\R)$ in all orders of $R^{-s}$ are   
obtained in terms of the on-shell scattering amplitudes $f^\pm(\q;\ka)$ only. This 
expansion together with obtained exact asymptotic expression for interference fluxes  
reveals for finite $R$ the necessity to replace the differential cross-section 
(\ref{3_}) by the normalized outgoing differential flux. 
Nevertheless, the second definition of Eq. (\ref{3_}) for total cross-section, which thus is 
replaced by total outgoing flux, remains unchanged together with the unitarity relation and 
the optical theorem, as all their asymptotic power corrections precisely disappear.

\section{Asymptotic expansion of scattering wave function}

To show the nature of asymptotic expansion we have to recall some properties 
\cite{GW,S,T,Ng,AR} of wave functions and amplitudes (\ref{2_}). The function  
$\Psi^\pm_\ka(\R)$ (\ref{1_}), being solution of Schr\"odinger equation for the energy 
$E>0$, satisfies Lippman-Schwinger equation:  
\begin{eqnarray}
&&\!\!\!\!\!\!\!\!\!\!\!\!\!\!\!\!\!\!\!
\left(\vec{\nabla}^2_{\R}+k^2\right)\Psi^\pm_\ka(\R)=V(R)\Psi^\pm_\ka(\R),
\,\mbox { for: }\, k^2=\frac{2M}{\hbar^2}E, \quad V(R)=\frac{2M}{\hbar^ 2}U(R),
\label{4_} \\
&&\!\!\!\!\!\!\!\!\!\!\!\!\!\!\!\!\!\!\! 
\Psi^\pm_\ka(\R)=e^{i(\ka\cdot\R)}-
\int\! d^3{\rm x}\,\frac{e^{\pm ik|\vec{\rm R}-\xv|}}{4\pi|\vec{\rm R}-\xv|}\,
V(|\xv|)\Psi^\pm_\ka(\xv)\equiv e^{i(\ka\cdot\R)}+{\cal J}^\pm_\ka(\R). 
\label{5_} 
\end{eqnarray}
Here the differential vector-operator and the operator of angular momentum square in 
the spherical basis $\n,\vec{\eta}_\vartheta,\vec{\eta}_\varphi$ have the following 
properties for $\vec{\rm R}=R\n$, 
\begin{eqnarray}
&&\!\!\!\!\!\!\!\!\!\!\!\!\!\!\!\!\!\!\!\!\!\!\!\! 
\n=(\sin\vartheta\cos\varphi,\sin\vartheta\sin\varphi,\cos\vartheta),\quad 
\vec{\eta}_\vartheta=\partial_\vartheta\n,\quad 
\vec{\eta}_\varphi \sin\vartheta =\partial_\varphi\n :  
\label{4} \\
&&\!\!\!\!\!\!\!\!\!\!\!\!\!\!\!\!\!\!\!\!\!\!\!\! 
\vec{\nabla}_{\vec{\rm R}}=\n\partial_R+\frac 1R \vec{\partial}_{\n}, \quad 
(\n\cdot\vec{\nabla}_{\vec{\rm R}})=\partial_R, \quad 
\vec{\partial}_{\n}=\vec{\eta}_\vartheta\partial_\vartheta+
\frac{\vec{\eta}_\varphi}{\sin\vartheta}\partial_\varphi, 
\label{3} \\
&&\!\!\!\!\!\!\!\!\!\!\!\!\!\!\!\!\!\!\!\!\!\!\!\! 
(\n\cdot\vec{\partial}_{\n})=0,\quad 
(\vec{\partial}_{\n}\cdot\n)=2,\quad 
(\n\times\vec{\partial}_{\n})^2=\vec{\partial}^2_{\n},\quad 
(\n\times\vec{\partial}_{\n})=i\vec{L}_{\n},
\label{5} \\
&&\!\!\!\!\!\!\!\!\!\!\!\!\!\!\!\!\!\!\!\!\!\!\!\! 
-\vec{\partial}^2_{\n}=\vec{L}^2_{\n}=
2R(\n\cdot\vec{\nabla}_{\vec{\rm R}})+R^2\left(
(\n\cdot\vec{\nabla}_{\vec{\rm R}})^2-\vec{\nabla}^2_{\vec{\rm R}}\right),\;
\mbox{ whence,}
\label{6} \\
&&\!\!\!\!\!\!\!\!\!\!\!\!\!\!\!\!\!\!\!\!\!\!\!\! 
\mbox{for }\;\cos\vartheta=c:\;\;
{\cal L}_{\n} \equiv \vec{L}^2_{\n}=
-\left[\partial_c(1-c^2)\partial_c+(1-c^2)^{-1}\partial^2_\varphi\right], 
\label{7}
\end{eqnarray} 
and the well-known representation also is used for arriving from point $\xv$ to point 
$\vec{\rm R}$ spherical wave being free 3-dimensional Green function: \cite{GW,S,T,Ng,AR} 
\begin{equation}
\frac{e^{\pm ik|\vec{\rm R}-\xv|} }{4\pi|\vec{\rm R}-\xv|}= 
\int\!\frac{d^3{\rm q}}{(2\pi)^3}\frac{e^{i(\q\cdot(\vec{\rm R}-\xv))}}
{(\q^2-k^2\mp i0)}.
\label{8}
\end{equation}
When $\xv=0$ it satisfies the well-known inhomogeneous equation: 
\begin{equation}
\left(\vec{\nabla}^2_{\vec{\rm R}}+k^2\right)
\frac{e^{\pm ik R}}{4\pi R}=-\,\delta_3(\vec{\rm R}).  
\label{9}
\end{equation}
Then the power index $\pm ikR$ is defined in the sense of analytic continuation 
with a small real negative admixture: \cite{Ng,AR} 
$\pm ik\mapsto -(-k^2\mp i0)^{1/2}=\pm ik -0$, which is almost nowhere written but 
is everywhere assumed. The following Lemma is in order. 

\begin {lemma} 
When $\vec{\rm R}=R\n$, $\xv=r\vs$, 
$\vs=(\sin\beta\cos\alpha, \sin\beta \sin\alpha,\cos\beta)$, 
$|\xv|=r<R$ and operator ${\cal L}_{\n}=\vec{L}^2_{\n}$ (or $\n\mapsto\vs$) is defined by 
Eqs. (\ref{3})--(\ref{7}) with positively defined operator 
${\cal L}_{\n}+\frac 14=(\Lambda_{\n}+\frac 12)^2$ such that 
$\Lambda_{\n}+\frac 12=\sqrt{{\cal L}_{\n}+\frac 14}$ is also positively defined, then: 
\begin{eqnarray}
&&\!\!\!\!\!\!\!\!\!\!\!\!\!\!\!\!\!\!\!\!\!\!\!\! 
\frac{e^{\pm ik|\vec{\rm R}-\xv|}}{4\pi|\vec{\rm R}-\xv|}=
\frac{\chi_{\Lambda_{\n}}(\mp ikR+0)}{4\pi R}e^{\mp ik(\n\cdot\xv)}\sim
\label{18_0} \\
&&\!\!\!\!\!\!\!\!\!\!\!\!\!\!\!\!\!\!\!\!\!\!\!\! 
\sim\frac{e^{\pm ikR}}{4\pi R} \left\{1+\sum^\infty_{s=1}\frac{ \displaystyle 
\prod\limits^s_{\mu=1}\left[{\cal L}_{\n}-\mu(\mu-1)\right]}
{s!(\mp 2ikR)^s}\right\} e^{\mp ik(\n\cdot\xv)}. 
\label{18} 
\end{eqnarray}
\end {lemma}
\begin {proof} 
The expression (\ref{18_0}) for $R>r$ is a formal operator rewriting of the 
usual multipole expansion of free Green function \cite{T} (\ref{8}) with the help 
of self-adjoint operator formally introduced instead of $l$: $l\mapsto\Lambda_{\n}$  
but never really arising and with the help of multipole expansion of plane wave, \cite{T} 
that are listed also in Ref.\cite{b_e2}, formulae (8.533), (8.534): 
\begin{eqnarray}
&&\!\!\!\!\!\!\!\!\!\!\!\!\!\!\!\!\!\!\!\!\!\!\!\! 
\frac{e^{\pm ik|\vec{\rm R}-\xv|}}{4\pi|\vec{\rm R}-\xv|}=
\frac{1}{4\pi kRr}\sum^\infty_{l=0}i^{\mp l}\chi_l(\mp ikR+0)\,\psi_{l\,0}(kr)
(2l+1)P_l\left((\n\cdot\vs)\right),
\label{19} \\
&&\!\!\!\!\!\!\!\!\!\!\!\!\!\!\!\!\!\!\!\!\!\!\!\! 
e^{\mp ik(\n\cdot\xv)}=
\frac{1}{kr}\sum^\infty_{l=0}i^{\mp l}\,\psi_{l\,0}(kr)
(2l+1)P_l\left((\n\cdot\vs)\right),
\label{20} 
\end{eqnarray} 
Here the spherical functions $Y^m_l(\n)=\langle\n|l\,m\rangle$ and Legendre 
polynomials $P_l(c)$ being eigenfunctions of self-adjoint operator (\ref{7}) on the unit 
sphere for $c=(\n\cdot\vs)$ or $c=\cos\vartheta$ satisfy the well-known orthogonality, 
parity, completeness and other conditions \cite{GW,S,T,Ng,AR,vil,b_l} 
(\ref{21}) -- (\ref{A_5}) with the delta-function $\delta_\Omega(\n,\,\vs)$ on the unit 
sphere. 

The solutions $\chi_l(\mp ikr)$, $\psi_{l\,0}(kr)$ of free radial Schr\"odinger equation:  
\begin{equation}
\left[r^2\left(\frac 1r\,\partial^2_r\,r+k^2\right)\right]\frac{\psi_{l\,0}(kr)}{r}=
l(l+1)\frac{\psi_{l\,0}(kr)}{r},  
\label{26} 
\end {equation}
are defined by Macdonald $K_\lambda(z)$ and Bessel $J_\lambda(y)$ functions 
\cite{b_e2,olv,vil,b_l} (\ref{A_1}) -- (\ref{A_10}) that for integer 
$l$, i.e. half integer $\lambda=l+\frac 12$ are reduced to elementary functions:  
\begin{equation}
\chi_l(b R)=\left(\frac{2bR}{\pi}\right)^{1/2}\!\!K_{l+{\frac 12}}(bR), \quad \;\;
\chi_l(b R)\! \underset{l=\mathrm{int}}{=\Longrightarrow}
e^{-b R} \sum^l_{s=0}\frac{(l+s)!}{s!(l-s)!(2b R)^s}. 
\label{24} 
\end {equation}
The function $K_\lambda(z)$ (\ref{A_1}) is entire function \cite{olv,b_e2} of $\lambda^2$. 
This is the reason the well-defined operator $\Lambda_{\n}$ introduced in Eq. (\ref{18_0}) 
does not appear explicitly. 

The expansion (\ref{18}) for large $R$ is the known asymptotic expansion of function 
(\ref{18_0}), being infinite asymptotic version \cite{b_e2,olv} of the sum (\ref{24}) 
for arbitrary non-integer $l$, $|{\rm arg}(bR)|<3\pi/2$, is supplemented by observation 
\cite{b_e2,olv} for the product:
\begin{equation}
\frac{(l+s)!}{(l-s)!}=\prod\limits^s_{\mu=1}(l-\mu+1)(l+\mu)=
\prod\limits^s_{\mu=1}\left[l(l+1)-\mu(\mu-1)\right].
\label{28}
\end{equation}
Due to (\ref{21}) it may be factored out \cite{k_t3} from the sum over $l$ (\ref{19}) as 
operator product in the right hand side of Eq. (\ref{18}), thus converting this sum into 
the expansion (\ref{18}). 
\end {proof} 

\noindent
{\sl Remark.} The operator ${\cal L}_{\n}$ in Eq. (\ref{18}) may be replaced by 
operator in square brackets of the left hand side of Eq. (\ref{26}) or by the similar 
operator with interchange of $r\rightleftharpoons k$ with the same result. 

\begin {theorem}
Let the potential $V(r)$ have finite first absolute moment and decrease at 
$r\to\infty$ faster than any power of $1/r$.   
Then the integral ${\cal J}^\pm_\ka(\vec{\rm R})$ in Eq. (\ref{5_}) for sufficiently 
large $R$ admits asymptotic power expansion whose coefficients are defined by the 
on-shell scattering amplitudes $f^\pm(\q;\ka)$ only. 
This expansion has asymptotic sense \cite{olv} even though the potential $V(r)$ in 
Eq. (\ref{4_}) has a finite support: 
\begin{eqnarray}
&&\!\!\!\!\!\!\!\!\!\!\!\!\!\!\!\!\!\!\!\!\!\!\!\! 
{\cal J}^\pm_\ka(\vec{\rm R})\sim \frac{e^{\pm ikR}}{R}\left\{f^\pm(k\n;\ka)
+\sum^\infty_{s=1}\frac{h^\pm_s(k\n;\ka)}{(\mp 2ikR)^s}\right\}, 
\label{29} \\
&&\!\!\!\!\!\!\!\!\!\!\!\!\!\!\!\!\!\!\!\!\!\!\!\! 
\mbox{with: }\;
h^\pm_s(k\n;\ka)=
\frac{1}{s!}\prod\limits^s_{\mu=1}\left[{\cal L}_{\n}-\mu(\mu-1)\right]f^\pm(k\n;\ka), 
\label{30} \\
&&\!\!\!\!\!\!\!\!\!\!\!\!\!\!\!\!\!\!\!\!\!\!\!\! 
\mbox{or: }\;
h^\pm_s(k\n;\ka)=\frac{{\cal L}_{\n}-s(s-1)}{s} h^\pm_{s-1}(k\n;\ka), \quad 
\ka=k\om,  
\label{30_0}
\end{eqnarray}
and is equivalent to infinite reordering of its asymptotic multipole 
expansion: \cite{T,Ng,AR}
\begin{eqnarray}
&&\!\!\!\!\!\!\!\!\!\!\!\!\!\!\!\!\!\!\!\!\!\!\!\! 
{\cal J}^\pm_\ka(\vec{\rm R})\simeq \frac{1}{R}
\sum^\infty_{j=0}\chi_j(\mp ikR)
(2j+1)\eta^\pm_j(k)P_j\left(\pm(\n\cdot\om)\right),
\label{31} \\
&&\!\!\!\!\!\!\!\!\!\!\!\!\!\!\!\!\!\!\!\!\!\!\!\! 
\mbox{with: }\;
h^\pm_s(k\n;\ka)=
\frac{1}{s!}\sum^\infty_{j=s}\frac{(j+s)!}{(j-s)!}
(2j+1)\eta^\pm_j(k)P_j\left(\pm(\n\cdot\om)\right),
\label{32} \\
&&\!\!\!\!\!\!\!\!\!\!\!\!\!\!\!\!\!\!\!\!\!\!\!\! 
\mbox{and: }\; 
f^\pm(k\n;\ka)=h^\pm_0(k\n;\ka)=
\sum^\infty_{j=0}(2j+1)\eta^\pm_j(k)P_j\left(\pm(\n\cdot\om)\right), 
\label{32_1} \\
&&\!\!\!\!\!\!\!\!\!\!\!\!\!\!\!\!\!\!\!\!\!\!\!\! 
\mbox{for: }\; 
f^\pm(k\n;\ka)=
-\,\frac{1}{4\pi}\int\! d^3{\rm x}\,e^{\mp ik(\n\cdot\xv)}V(|\xv|)\Psi^\pm_\ka(\xv), 
\quad \xv=r\vs,
\label{32_2} 
\end{eqnarray}
as the usual on-shell scattering amplitude \cite{GW,S,T,Ng,AR}.
\end {theorem}
\begin {proof} 
Suppose at first the finite support for $V(r)$ at $r\leqslant a$. Then for $R>a$ we  
can directly substitute the expression (\ref{18_0}) into representation (\ref{5_}) for  
${\cal J}^\pm_\ka(\vec{\rm R})$ with the following result after interchange of the order 
of differentiation and integration for Fourier transformation (\ref{32_2}), what is justified 
\cite{olv} also for asymptotic series (\ref{18}): 
\begin{eqnarray}
&&\!\!\!\!\!\!\!\!\!\!\!\!\!\!\!\!\!
{\cal J}^\pm_\ka(\vec{\rm R})\simeq
\frac{{\chi}_{\Lambda_{\n}}(\mp ikR)}{R}\,f^\pm(k\n;\ka) \sim 
\label{33_0} \\
&&\!\!\!\!\!\!\!\!\!\!\!\!\!\!\!\!\!
\sim \frac{e^{\pm ikR}}{R}\left\{1+\sum^\infty_{s=1}\frac{ \displaystyle 
\prod\limits^s_{\mu=1}\left[{\cal L}_{\n}-\mu(\mu-1) \right]}
{s!(\mp 2ikR)^s}\right\}f^\pm(k\n;\ka).
\label{33}
\end{eqnarray}
This is exactly the asymptotic expansion (\ref{29}) with the coefficients $h^\pm_s(k\n;\ka)$  
defined by Eq. (\ref{30}). However, that is not the case for the potential $V(r)$ with 
infinite support. 
Estimating it for $r>R$ as $|V(r)|<{\rm C}_{\rm N}/r^{\rm N}$ with arbitrary finite 
${\rm N}\gg 1$, two pieces of correction that should be added may be easy estimated as:  
\begin{eqnarray}
&&\!\!\!\!\!\!\!\!\!\!\!\!\!\!\!\!\!
\Delta_R{\cal J}^\pm=-\int\limits_{r>R}\!d^3{\rm x}\,
\frac{e^{\pm ik|\vec{\rm R}-\xv|}}{4\pi|\vec{\rm R}-\xv|}\,V(r)\Psi^\pm_\ka(\xv),
\label{34_0} \\
&&\!\!\!\!\!\!\!\!\!\!\!\!\!\!\!\!\!
\Delta_R f^\pm= \frac{{\chi}_{\Lambda_{\n}}(\mp ikR)}{4\pi R}\! 
\int\limits_{r>R}\!d^3{\rm x}\,e^{\mp ik(\n\cdot\xv)}V(r)\Psi^\pm_\ka(\xv),
\label{34} \\
&&\!\!\!\!\!\!\!\!\!\!\!\!\!\!\!\!\!
|\Delta_R{\cal J}^\pm|<\frac{||\Psi|| {\rm C}_{\rm N}}{({\rm N}-2)R^{{\rm N}-2}},
\quad \;\;
|\Delta_R f^\pm|<\frac{||\Psi||{\rm C}_{\rm N}}{({\rm N}-3)R^{{\rm N}-2}}
\left[1+ O(R^{-1})\right], 
\label{34_1} \\
&&\!\!\!\!\!\!\!\!\!\!\!\!\!\!\!\!\!
\mbox{with the finite norm \cite{Ng,AR}: 
$||\Psi||=\underset{\xv}{\sup}|\Psi(\xv)|$ of functions }\;\Psi^\pm_\ka(\xv). 
\label{34_2}
\end{eqnarray}
Due to the arbitrariness of ${\rm N}\gg 1$ for these corrections the asymptotic expansion 
conserves its form (\ref{29}), (\ref{33}) but acquires additional asymptotic sense  
\cite{olv} compared with expansion (\ref{18}). 

Indeed, due to the partial wave decomposition (\ref{32_1}) of scattering amplitude 
$f^\pm(k\n;\ka)$, expression (\ref{33_0}) is the formal operator rewriting of   
asymptotic multipole expansion \cite{T,Ng,AR} (\ref{31}) of 
${\cal J}^\pm_\ka(\vec{\rm R})$. Unlike its exact expression given by Eq. (\ref{5_})  
the expansion (\ref{31}) according to Eqs. (\ref{26}), (\ref{24}) is a solution of free 
Schr\"odinger equation like Eq. (\ref{9}) with $R>0$. 
When $V(r)=0$ at $r>a$, the Schr\"odinger operators in (\ref{4_}) and  (\ref{9}) 
coincide for $R>a$. Then both asymptotic relations (\ref{31}), (\ref{33_0}) 
become exact expressions due to convergence of the expansion (\ref{31}) for $R>a$  
in the usual sense \cite{T,Ng,AR} similarly to expansions (\ref{19}), (\ref{20}). 
At the same time, the expansion (\ref{29}), i.e. (\ref{33}), conserves its asymptotic 
sense acquired according to Lemma 1. 

The assumed potential $V(r)$ for the case of infinite support has only finite effective 
radius \cite{Ng} and provides a slowdown fall \cite{Ng,AR} of partial waves at 
$j\to\infty$, e.g. like $|\eta_j(k)|\sim e^{-\tau j}$, $\tau>0$ for potential of Yukawa-type. 
This is enough for convergence of partial wave decompositions (\ref{32}), (\ref{32_1}) 
but can not provide convergence of the multipole expansion (\ref{31}) which now 
also acquires the asymptotical sense. 
Its infinite reordering (\ref{31}) $\mapsto$ (\ref{29}), (\ref{32}) given here 
simply ``displaces'' this asymptotic sense from the summation over angular momentum 
$j$ onto the always asymptotic expansion on integer powers $R^{-s}$ whose 
coefficients now are well-defined as derivatives (\ref{30}) of scattering 
amplitude with respect to $c=(\n\cdot\om)$, or as convergent partial wave decompositions 
(\ref{32}). Thus, all these coefficients are observable. 
\end {proof} 
\noindent
{\sl Remark.} From the estimations (\ref{34_0}) -- (\ref{34_2}) it is also clear 
\cite{T,Ng,AR} that even standard asymptotic (\ref{1_}) requires for the potential 
${\rm N}>3$ at least. More generally these estimations mean that for 
$|V(r)|\leqslant {\rm C}_{\rm N}/r^{\rm N}$ with $r\to\infty$ the asymptotic expansion 
(\ref{33}) is applied until $s\leqslant[{\rm N}-3]$. Thus, the further consideration is 
possible only for potentials $V(r)$, specified with the conditions of Theorem 1.

\section{Differential fluxes and unitarity relation}

To make a careful analysis of different fluxes, the following Lemma 2 is useful.
\begin {lemma} 
The function $e^{ikr[1-(\n\cdot\vs)]}$ as a distribution on the space of infinitely 
smooth functions ${\cal H}(\n)$ on the unit sphere $\n(\cos\vartheta,\varphi)$,  
parametrized by (\ref{4}), has the following exact operator representation for 
$c=(\n\cdot\vs)$. Let $\overline{\cal H}(c)$ be defined as:  
\begin{eqnarray}
&&\!\!\!\!\!\!\!\!\!\!\!\!\!\!\!\!\!
\overline{\cal H}(c)=\int\limits^{2\pi}_0\! d\varphi\, {\cal H}(\n(c,\varphi)),
\;\;\;\mbox{ then:}
\label {S_00} \\
&&\!\!\!\!\!\!\!\!\!\!\!\!\!\!\!\!\!
\int\!d\Omega(\n)\,e^{ikr[1-(\n\cdot\vs)]}\,{\cal H}(\n)\equiv 
\int\limits^1_{-1}\!dc\,e^{ikr(1-c)}\,\overline{\cal H}(c)=
\label {S_0} \\
&&\!\!\!\!\!\!\!\!\!\!\!\!\!\!\!\!\!
=\int\limits^1_{-1}\!dc\left[\delta(1-c)-e^{2ikr}\delta(1+c)\right]
\left(-ikr+\partial_c\right)^{-1}\overline{\cal H}(c).
\label {S_1}
\end{eqnarray}
\end {lemma} 
\begin {proof}
With $d\Omega(\n)=\sin\vartheta\,d\vartheta\, d\varphi=-dc\,d\varphi$ the result is 
obtained by using integration over $c$ by parts infinite number of times. The operator in 
Eq. (\ref{S_1}) has a sense of a formal series over powers of differential operator 
$\partial_c$. The well-known standard asymptotic relation \cite{GW,S,T,Ng} of the first 
order on $1/r$ corresponds here to $\partial_c\mapsto 0$. 
\end {proof}

Now let's consider the elementary flux of non-diagonal current 
$\jv_{\q,\ka}(\R)$ through a small element of spherical surface $\n R^2 d\Omega(\n)$, 
for $\R=R\n$, $\q=k\vs$, $\ka=k\om$, and $\overset{\leftrightarrow}{\vec{\nabla}}_{\R}=
\overset{\rightarrow}{\vec{\nabla}}_{\R}-\overset{\leftarrow}{\vec{\nabla}}_{\R}$, 
$\overset{\leftrightarrow}{\partial}_R=
\overset{\rightarrow}{\partial}_R-\overset{\leftarrow}{\partial}_R=
\bigl(\n\cdot\overset{\leftrightarrow}{\vec{\nabla}}_{\R}\bigr)$ according to 
(\ref{3}) -- (\ref{6}). Total flux through any closed surface is zero because the 
current is conserved \cite{Ng} due to Eq. (\ref{4_}): 
\begin{eqnarray}
&&\!\!\!\!\!\!\!\!\!\!\!\!\!\!\!\!\!
\jv_{\q,\ka}(\R)=\frac 1{2i}\left[\left(\Psi^+_\q(\R)\right)^*  
\overset{\leftrightarrow}{\vec{\nabla}}_{\R}\Psi^+_\ka(\R)\right], \qquad 
\left(\vec{\nabla}_{\R}\cdot\jv_{\q,\ka}(\R)\right)=0,
\label {S_2} \\
&&\!\!\!\!\!\!\!\!\!\!\!\!\!\!\!\!\!
R^2 d\Omega(\n)\left(\n\cdot\jv_{\q,\ka}(\R)\right)=R^2 d\Omega(\n)
\frac 1{2i}\left[\left(\Psi^+_\q(\R)\right)^*\overset{\leftrightarrow}{\partial}_R
\Psi^+_\ka(\R)\right]\longmapsto
\label {S_3} \\
&&\!\!\!\!\!\!\!\!\!\!\!\!\!\!\!\!\!
\longmapsto R^2 d\Omega(\n)\,\frac{k}2\, e^{ikR(\n\cdot(\om-\vs))}(\n\cdot(\om+\vs))+
\label {S_4} \\
&&\!\!\!\!\!\!\!\!\!\!\!\!\!\!\!\!\!
+\frac {d\Omega(\n)}{2i}
\left[\overset{*}{f}{}^+(k\n;k\vs)
\left(\chi_{\overset{\leftarrow}{\Lambda}{\!}_{\n}}(ikR)
\overset{\leftrightarrow}{\partial}_R 
\chi_{\overset{\rightarrow}{\Lambda}{\!}_{\n}}(-ikR)\right)f^+(k\n;k\om)\right]-
\label {S_5} \\
&&\!\!\!\!\!\!\!\!\!\!\!\!\!\!\!\!\!
-\frac {d\Omega(\n)}{2i}
\left\{\!\left(\!e^{ikR[1-(\n\cdot\vs)]}\,
e^z\!\left[z(\n\cdot\vs)+1-z\frac{\partial}{\partial z}\right]\!
\chi_{\overset{\rightarrow}{\Lambda}{\!}_{\n}}(z)f^+(k\n;k\om)\right)_{\!z\,=\,0-ikR}
\!\!\!\!\!\!\!-\right.
\label {S_6} \\
&&\!\!\!\!\!\!\!\!\!\!\!\!\!\!\!\!\!
\left. -\,\biggl(\vs \rightleftharpoons\om\biggr)^*\right\}.
\label {S_7}
\end{eqnarray}
Here, for sufficiently large $R$, the expressions (\ref{5_}) and (\ref{33_0}) for the 
wave function $\Psi^+_\ka(\R)$ were used. As well as in Eqs. (\ref{S_2}), (\ref{S_3}), 
the arrows point out the directions of action for the operators  
$\overset{\leftarrow}{\Lambda}{\!}_{\n}$ and $\overset{\rightarrow}{\Lambda}{\!}_{\n}$
from Lemma 1, that in fact are directions of action for the operators 
$\overset{\leftarrow}{\cal L}{\!}_{\n}$ and $\overset{\rightarrow}{\cal L}{\!}_{\n}$ 
(\ref{6}). Integration of separate terms over solid angle $d\Omega(\n)$ with fixed $R$ 
gives here the following interesting results. For the flux of incoming plane waves 
(\ref{S_4}), since $\left((\om-\vs)\cdot(\om+\vs)\right)=\om^2-\vs^2=0$, one has: \cite{Ng} 
\begin{equation}
R^2\int\!d\Omega(\n)\,\frac{k}2\, e^{ikR(\n\cdot(\om-\vs))}(\n\cdot(\om+\vs))=0. 
\label {S_8}
\end{equation}
For the flux (\ref{S_5}) we can ignore the arrows of $\Lambda_{\n}$ because operator 
${\cal L}{\!}_{\n}$ (\ref{7}) is self-adjoint on the unit sphere. So, the Wronskian 
(\ref{A_9}) leads to the total flux: 
\begin{eqnarray}
&&\!\!\!\!\!\!\!\!\!\!\!\!\!\!\!\!\!
\int\! \frac {d\Omega(\n)}{2i}\left[\overset{*}{f}{}^+(k\n;k\vs)
\left(\chi_{\overset{\leftarrow}{\Lambda}{\!}_{\n}}(ikR)
\overset{\leftrightarrow}{\partial}_R 
\chi_{\overset{\rightarrow}{\Lambda}{\!}_{\n}}(-ikR)\right)f^+(k\n;k\om)\right]=
\label {S_9}\\
&&\!\!\!\!\!\!\!\!\!\!\!\!\!\!\!\!\!
=k \int\! d\Omega(\n)\overset{*}{f}{}^+(k\n;k\vs)f^+(k\n;k\om).
\nonumber 
\end{eqnarray}
This is the total non-diagonal outgoing flux for finite $R$, obtained from the line 
(\ref{S_5}), now taking into account all possible asymptotic power corrections. 
Nevertheless, it looks exactly like right hand side of unitarity relation \cite{GW,S,T,Ng,AR}
independent of $R$. 
It is clear the same result may be obtained using the partial wave decomposition 
(\ref{32_1}) with the help of Eqs. (\ref{A_2}), (\ref{A_9}) (clf. (\ref{S_19_0}), 
(\ref{S_24}), (\ref{S_25}) below).   

The lines (\ref{S_6}), (\ref{S_7}) represent the non-diagonal interference ($\vs\neq\om$)
between incoming and outgoing fluxes. According to Lemma 2 for the first 
exponential of (\ref{S_6}), it takes place only in corresponding forward and backward 
directions. Note that any averaging over $R$ due to rapidly oscillating exponent 
$e^{2ikR}$ eliminates \cite{T} the contribution of backward direction in Eq. (\ref{S_1}). 
With this elimination and definition (\ref{S_00}), the line (\ref{S_6}) for 
$(\n\cdot\vs)=c$ gives: 
\begin{eqnarray}
&&\!\!\!\!\!\!\!\!\!\!\!\!\!\!\!\!\!
-\int\limits^{2\pi}_0\!\frac {d\varphi}{2i}\!\int\limits^1_{-1}\!dc\left(\delta(1-c) 
\frac{e^z}{(z+\partial_c)}\left[zc+1-z\frac{\partial}{\partial z}\right]
\chi_{\overset{\rightarrow}{\Lambda}{\!}_{\n}}(z)f^+(k\n;k\om)\right)_{z\,=\,0-ikR}.
\label {S_11} 
\end{eqnarray}
By moving the operator from denominator into the exponential for $z=0-ikR$:
\begin{equation}
\frac{e^z}{(z+\partial_c)}= \int\limits^\infty_0\!d\xi e^{z(1-\xi)-\xi\partial_c},
\label {S_12}
\end{equation}
after simple commutations, one obtains for Eq. (\ref{S_11}) 
\begin{eqnarray}
&&\!\!\!\!\!\!\!\!\!\!\!\!\!\!\!\!\!
=- \frac z{2i}\!\int\limits^{2\pi}_0\!d\varphi \!\int\limits^\infty_0\!d\xi \,
e^{-\xi\partial_c}\left[\chi_{\overset{\rightarrow}{\Lambda}{\!}_{\n}}(z)
\left(\overset{\leftrightarrow}{\partial}_z +\frac 1z \right)e^{z(1-\xi)}\right]
f^+(k\n;k\om)\biggr|^{c=1}_{z\,=\,0-ikR}, 
\label {S_13} 
\end{eqnarray}
with $c=1$ where possible. 
For the arbitrary term of partial wave decomposition (\ref{32_1}) the scattering 
amplitude here is effectively replaced by Legendre polynomial: 
$f^+(k\n;k\om)\mapsto P_j\left((\n\cdot\om)\right)$. This substitution immediately replaces 
$\chi_{\overset{\rightarrow}{\Lambda}{\!}_{\n}}(z)\mapsto \chi_j(z)$, thus 
permitting to make all remaining $\varphi$- integration and $\partial_
c$- differentiations in the closed form by using the relations (\ref{A_3}), (\ref{A_6}), 
(\ref{A_9}), wherefrom for: 
\begin{equation}
\int\limits^{2\pi}_0\!d\varphi P_j\left((\n\cdot\om)\right)=
2\pi P_j\left((\vs\cdot\om)\right) P_j(c), \qquad 
e^{-\xi\partial_c}P_j(c)\bigr|_{c=1}=P_j(1-\xi), 
\label {S_17} 
\end{equation}
it follows: 
\begin{eqnarray}
&&\!\!\!\!\!\!\!\!\!\!\!\!\!\!\!\!\!
-\frac z{2i}\!\int\limits^{2\pi}_0\!d\varphi\!\int\limits^\infty_0\!d\xi\, 
e^{-\xi\partial_c}\left[
\chi_j(z)\left(\overset{\leftrightarrow}{\partial}_z +\frac 1z \right)e^{z(1-\xi)}\right]
P_j\left((\n\cdot\om)\right)\biggr|^{c=1}_{z\,=\,0-ikR}=
\label {S_14} \\
&&\!\!\!\!\!\!\!\!\!\!\!\!\!\!\!\!\!
=-\,\frac {2\pi}{2i}\,z\,P_j\left((\vs\cdot\om)\right)\left[\chi_j(z)
\left(\overset{\leftrightarrow}{\partial}_z+\frac 1z \right)
\!\int\limits^\infty_0\!d\xi\, e^{z(1-\xi)}\,e^{-\xi\partial_c}P_j(c)
\right]^{c=1}_{z\,=\,0-ikR}\!\!=
\label {S_15} \\
&&\!\!\!\!\!\!\!\!\!\!\!\!\!\!\!\!\!
=-\,\frac {2\pi}{2i}\,z\,P_j\left((\vs\cdot\om)\right)\left[\chi_j(z)
\left(\overset{\leftrightarrow}{\partial}_z+\frac 1z \right)\frac{\chi_j(-z)}{z}\right]=
-\,\frac {4\pi}{2i}P_j\left((\vs\cdot\om)\right),   
\label {S_16} \\
&&\!\!\!\!\!\!\!\!\!\!\!\!\!\!\!\!\!
\mbox{or, equivalently: (\ref{S_11}) = (\ref{S_13})}=
-\,\frac {4\pi}{2i}f^+(k\vs;k\om).
\label {S_18}
\end{eqnarray}
Thus, the contribution of the lines (\ref{S_6}), (\ref{S_7}) into the full flux 
in accordance with the left hand side of unitarity condition \cite{GW,S,T,Ng,AR} becomes 
equal to: 
\begin{equation}
-\,\frac {4\pi}{2i}\left[f^+(k\vs;k\om)- \overset{*}{f}{}^+(k\om;k\vs)\right]=
-4\pi\,Im\,f^+(k\vs;k\om),  
\label {S_19}
\end{equation}
but now taking into account all the possible asymptotic power corrections:  
\begin{equation}
\frac{4\pi}k \,Im\,f^+(k\vs;k\om)=
\int\!d\Omega(\n)\overset{*}{f}{}^+(k\n;k\vs)f^+(k\n;k\om).  
\label {S_19_0}
\end{equation}
The diagonal case $\vs=\om$ in Eqs.  (\ref{S_9}), (\ref{S_19_0}) represents the optical 
theorem \cite{GW,S,T,Ng,AR} with total cross-section $\sigma$ of Eq. (\ref{3_}) 
in the right hand side and is not changed by these power corrections also. 
Moreover, since the operator of angular momentum square  (\ref{7}) depends for this case 
only on one variable: ${\cal L}_{\n}\mapsto \partial_c(c^2-1)\partial_c$, the result 
(\ref{S_18}) may be checked in first several orders of $R^{-s}$ directly from Eqs. 
(\ref{33_0}), (\ref{33}), (\ref{S_11}) on operator level. 

However, for finite $R$ the differential cross-section of Eq. (\ref{3_}) has to be replaced 
now by diagonal outgoing differential flux $\widehat{d\sigma}(R)$ (\ref{S_5}) normalized 
\cite{GW,S,T,Ng,AR} to the density $k$ of incoming flux (\ref{S_4}) for $\vs=\om$. It still 
contains asymptotic power corrections defined by Eqs. (\ref{33_0}), (\ref{33}) of Theorem 1: 
\begin{eqnarray}
&&\!\!\!\!\!\!\!\!\!\!\!\!\!\!\!\!\!
\frac{\widehat{d\sigma}(R)}{d\Omega(\n)}=\frac{1}{2ik} \left[\overset{*}{f}{}^+(k\n;k\om)
\left(\chi_{\overset{\leftarrow}{\Lambda}{\!}_{\n}}(ikR)
\overset{\leftrightarrow}{\partial}_R 
\chi_{\overset{\rightarrow}{\Lambda}{\!}_{\n}}(-ikR)\right)f^+(k\n;k\om)\right]= 
\label {S_20} \\
&&\!\!\!\!\!\!\!\!\!\!\!\!\!\!\!\!\!
=|f^+(k\n;k\om)|^2-\frac 1{kR}\,Im\left[\overset{*}{f}{}^+(k\n;k\om)
\overset{\rightarrow}{\cal L}{}_\n f^+(k\n;k\om)\right]+\frac 1{4(kR)^2}\cdot
\nonumber \\
&&\!\!\!\!\!\!\!\!\!\!\!\!\!\!\!\!\!
\cdot\left\{\left|\overset{\rightarrow}{\cal L}{}_\n f^+(k\n;k\om)\right|^2 -
Re\left[\overset{*}{f}{}^+(k\n;k\om)\overset{\rightarrow}{\cal L}{}^2_\n
f^+(k\n;k\om)\right]\right\}+O\left(\frac 1{R^3}\right).
\nonumber 
\end{eqnarray}
In terms of partial wave decomposition (\ref{32_1}) with corresponding phase 
shifts $\delta_j(k)$, for $c=(\n\cdot\om)$, $\eta_j(k)=\eta^+_j(k)$, 
$k\eta_j(k)=e^{i\delta_j(k)}\sin\delta_j(k)$, $\Delta_{jl}=j(j+1)-l(l+1)$, 
with the help of (\ref{A_10}) -- (\ref{A_8}), it reads: 
\begin{eqnarray}
&&\!\!\!\!\!\!\!\!\!\!\!\!\!\!\!\!\!
\frac{\widehat{d\sigma}(R)}{d\Omega(\n)}=
\sum^\infty_{l,j=0}(2l+1)(2j+1)\overset{*}{\eta}{}_l(k)\eta_j(k)
P_l(c)P_j(c)\,
\frac{(\chi_l(ikR)\overset{\leftrightarrow}{\partial}_R \chi_j(-ikR))}{2ik}, 
\label {S_24} \\ 
&&\!\!\!\!\!\!\!\!\!\!\!\!\!\!\!\!\!
\mbox{where: }\;
\frac{(\chi_l(ikR)\overset{\leftrightarrow}{\partial}_R \chi_j(-ikR))}{2ik}= 
1-\frac{\Delta_{jl}}{2ikR}-\frac{\Delta^2_{jl}}{8(kR)^2}+
O\left(\frac 1{R^3}\right). 
\label {S_25} 
\end{eqnarray}
The power corrections arising in (\ref{S_20}) or in this two-fold series for $j\neq l$  
in (\ref{S_25}), may be observable for slowly moving particles with $k\to 0$. 
They contain only real or imaginary parts of the products 
$\overset{*}{\eta}{}_l(k)\eta_j(k)$ and automatically disappear for $j=l$ in the total 
outgoing flux $\widehat{\sigma}$ being the total cross-section (\ref{3_}) now, 
or in the limit $R\to\infty$ for outgoing differential flux: 
\begin{equation}
\sigma = \widehat{\sigma} \equiv  
\int\!d\Omega(\n)\,\frac{\widehat{d\sigma}(R)}{d\Omega(\n)}, \qquad 
d\sigma=\lim_{R\to\infty}\widehat{d\sigma}(R). 
\label {S_25_0} 
\end{equation}
Since for real potential the Born approximation for amplitudes $f^+(k\n;k\om)$, 
$\eta_j(k)$ is real, \cite{GW,S,T,Ng} it is not enough to obtain the non zero first order 
correction $R^{-1}$ from Eqs. (\ref{S_20}) -- (\ref{S_25}). 
The relations (\ref{33_0}), (\ref{S_9}), 
(\ref{S_20}) -- (\ref{S_25_0}) and the observed disappearance of asymptotic corrections 
to Eq. (\ref{S_19_0}) are the main results of this work. 

\section{Identical particles with spin}    

In case of mutual scattering of identical Bose- or Fermi-particles \cite{GW,S,T} 
with spin $S$ one faces symmetrical or antisymmetrical scattering wave functions, 
amplitudes and respective cross-sections. The proper generalizations are 
straightforward, and instead of (\ref{5_}), (\ref{33_0}) one has: 
\begin{eqnarray}
&&\!\!\!\!\!\!\!\!\!\!\!\!\!\!\!\!\!
\Psi^+_{\ka(\pm)}(\R)\simeq  e^{i(\ka\cdot\R)}\pm e^{-i(\ka\cdot\R)}+ 
\frac{ \chi_{\overset{\rightarrow}{\Lambda}{\!}_{\n}}(-ikR) }{R}\,
{\rm F}^+_{(\pm)}(k\n;\ka), 
\label {S_26} \\ 
&&\!\!\!\!\!\!\!\!\!\!\!\!\!\!\!\!\!
\mbox {with: }\;
{\rm F}^+_{(\pm)}(k\n;\ka)=f^+(k\n;\ka)\pm f^+(-k\n;\ka),\;\; \mbox { and then:}
\label {S_27} \\ 
&&\!\!\!\!\!\!\!\!\!\!\!\!\!\!\!\!\!
\frac{\widehat{d\sigma}_{(\pm)}(R)}{d\Omega(\n)}=\frac{1}{2ik} 
\left[\overset{*}{{\rm F}}{}^+_{(\pm)}(k\n;k\om)
\left(\chi_{\overset{\leftarrow}{\Lambda}{\!}_{\n}}(ikR)
\overset{\leftrightarrow}{\partial}_R 
\chi_{\overset{\rightarrow}{\Lambda}{\!}_{\n}}(-ikR)\right)
{\rm F}^+_{(\pm)}(k\n;k\om)\right]. 
\label {S_28} 
\end{eqnarray}
Of course, the partial wave decompositions like (\ref{32_1}), (\ref{S_24}) contain now 
only even $j$ for ${\rm F}^+_{(+)}$ and only odd $j$ for ${\rm F}^+_{(-)}$. 
The normalized outgoing differential fluxes (\ref{S_28}) again have to replace 
corresponding differential cross-sections. For the scattering of nonpolarized identical 
particles the outgoing differential fluxes are defined by the usual way \cite{S} as: 
\begin{equation}
\widehat{d\sigma}_S(R)=
{\rm w}_{(+)}(S)\,\widehat{d\sigma}_{(+)}(R)+{\rm w}_{(-)}(S)\,\widehat{d\sigma}_{(-)}(R), 
\label {S_29}
\end{equation}
with the well known \cite{GW,S} probabilities ${\rm w}_{(\pm)}(S)$ for Bose- and 
Fermi-particles. Similarly (\ref{S_25_0}) integration of the flux (\ref{S_29}),  
(\ref{S_28}) over solid angle again obviously leads to the independent of $R$ total 
cross-section $\sigma_S$ for identical particles with spin $S$. 

\section{Conclusions}

As it is well-known, for a point-like, even anisotropic, stationary source of classical 
particles, or rays of light, or incompressible fluid, the radial flux of outgoing particles 
in a given solid angle does not depend on distance $R$ at all, due to the local 
conservation of classical current density. 

Turning to the wave picture, such independence is true only for the flux of pure spherical 
outgoing or incoming wave in Eq. (\ref{1_}) (see the first term in the right hand side of 
Eq. (\ref{S_20})).  That all results in well-known inverse-square law for event rate, 
which explicitly contains $1/R^2$ (see, e.g. Refs. \cite{N_shk,NN_shk}). 
A possible violation of this law is the subject of our interest. 
We show this violation is a pure wave effect, arising from nonsphericity of the 
exact scattering wave, i.e. from the next terms $R^{-s}$  $(s>1)$ of asymptotic expansion. 
The last is investigated here up to all orders, also using again the conservation of 
corresponding current.  

To this end, using the operator-valued asymptotic expansion for free Green function of  
Helmholtz equation, the asymptotic expansion for the wave function of potential 
scattering on inverse integer power of distance $R$ from scattering center is obtained. 
It is shown how these power corrections affect the definition of outgoing differential 
flux and interference flux.

Surprisingly, these power corrections precisely entirely disappear in total outgoing flux, 
unitarity relation, and optical theorem due to integration over solid angle at finite $R$. 
Thus, the applicability domain of these relations\footnote{Calculation similar to  
(\ref{S_11})--(\ref{S_16}) reveals exact disappearance of the contribution of backward 
direction in Eq. (\ref{S_1}) even without averaging over $R$ because 
$(\chi_j(z)\overset{\leftrightarrow}{\partial}_z\chi_j(z))\equiv 0$.} 
naturally extends to finite $R$ for fast ehough decrising potentials. 

It is worth to note that all obtained corrections are defined by observable on-shell 
amplitude or partial phase shifts. 
Nevertheless, the real observation of this dependence involves reevaluation of the phase 
shifts extracted earlier in fact from outgoing differential flux at finite $R$ 
(\ref{S_20}), (\ref{S_24}), (\ref{S_29}) without taking into account any corrections on the  
finite distance. 

Although asymptotic expansion by its nature has no sense as infinite sum, 
the obtained asymptotic expansions of the wave function and outgoing differential flux have 
a sense up to any finite order $s$ of $R^{-s}$ if potential $U(R)$ has a finite support or 
decreases for $R\to\infty$ faster than any power of $1/R$. Otherwise the maximal order 
$s$ of their validity is governed directly by the potential according to the Remark for  
Theorem 1. For example, the first two corrections, given by Eq. (\ref{S_20}), i.e.  
Eqs. (\ref{S_24}), (\ref{S_25}), may be applied to potentials with ${\rm N}>5$. 
Their disappearence on integration over solid angle in Eq. (\ref{S_25_0}) 
obviously takes place separatly in each order of $R^{-s}$.

So, following the authors of Ref.\cite {Grimus}, we come to conclusion, that   
``... the physics under consideration seems to comply naturally with the 
mathematical requirements''. 

\appendix 
\section{}  
The following relations for spherical functions and Legendre polynomials  
are necessary \cite{T,Ng,AR,b_e2,olv,vil,b_l} for $\n^2=\vs^2=\om^2=1$, with 
$\n(\cos\vartheta,\varphi)$, $\vs(\cos\beta,\alpha)$ parametrized by Eq. (\ref{4}) and 
Lemma 1, and $c$ may be equal to any of values 
$\cos\vartheta,\,(\n\cdot\vs),\,(\n\cdot\om)$:  
\begin{eqnarray}
&&\!\!\!\!\!\!\!\!\!\!\!\!\!\!\!\!\!\!\!\!\!\!\!\! 
\vec{L}^2_{\n}Y^m_l(\n)=l(l+1)Y^m_l(\n), \qquad
\vec{L}^2_{\n}P_l(c)=l(l+1)P_l(c),
\label{21} \\
&&\!\!\!\!\!\!\!\!\!\!\!\!\!\!\!\!\!\!\!\!\!\!\!\! 
\int\!d\Omega(\n)\, \overset{*}{Y}\!{}^m_l(\n)Y^{m^\prime}_j(\n)=
\delta_{lj}\delta_{mm^\prime}, \quad 
\sum^\infty_{l=0}\frac{(2l+1)}{4\pi}P_l\left((\n\cdot\vs)\right)=
\delta_\Omega(\n\,,\vs), 
\label{22_0} \\
&&\!\!\!\!\!\!\!\!\!\!\!\!\!\!\!\!\!\!\!\!\!\!\!\!  
\sum^l_{m=-l}\! Y^m_l(\n)\overset{*}{Y}\!{}^m_l(\vs) = 
\frac{(2l+1)}{4\pi}P_l\left((\n\cdot\vs)\right), \quad \; (-1)^lY^m_l(\n)=Y^m_l(-\n). 
\label{22} \\
&&\!\!\!\!\!\!\!\!\!\!\!\!\!\!\!\!\!\!\!\!\!\!\!\! 
Y^{m}_j(\vec{\rm e}_3)=i^j \sqrt{\frac{2j+1}{4\pi}}\,\delta_{m0}, \quad 
Y^{0}_j(\n)=i^j\sqrt{\frac{2j+1}{4\pi}}\,P_j({\rm n}_3=\cos\vartheta), 
\label{A_4} \\
&&\!\!\!\!\!\!\!\!\!\!\!\!\!\!\!\!\!\!\!\!\!\!\!\! 
\mbox {wherefrom: }\;
\int\!d\Omega(\n)\,P_l\left((\vs\cdot\n)\right)P_j\left((\n\cdot\om)\right)=
\frac{4\pi\,\delta_{lj}}{(2j+1)}P_j\left((\vs\cdot\om)\right), 
\label{A_2} \\
&&\!\!\!\!\!\!\!\!\!\!\!\!\!\!\!\!\!\!\!\!\!\!\!\! 
P_l\left(\cos\beta\right)P_l\left(\cos\vartheta\right)=
\frac 1{2\pi}\int\limits^{2\pi}_0\!d\varphi 
P_l\left(\cos\beta\cos\vartheta+\sin\beta\sin\vartheta\cos(\varphi-\alpha)\right),   
\label{A_3} \\
&&\!\!\!\!\!\!\!\!\!\!\!\!\!\!\!\!\!\!\!\!\!\!\!\! 
\mbox{and besides: }\;\; 
P_l(1-\xi)= \sum^l_{s=0}\frac{(l+s)!}{(l-s)!}\frac{(-\xi)^s}{(s!)^2 2^s}. 
\label{A_5} 
\end{eqnarray} 
Macdonald and Bessel functions\cite{b_e2,olv,vil} in definitions (\ref{26}), (\ref{24}) 
are defined by the relations, with $|{\rm arg}\,u-\beta_{1,2}|<\pi/2$, as:  
\begin{eqnarray}
&&\!\!\!\!\!\!\!\!\!\!\!\!\!\!\!\!\!\!\!\!
K_\lambda(u)=\frac 12 \int\limits^{\infty e^{-i\beta_1}}_{0 e^{i\beta_2}}
\frac{dt}{t}t^{\pm\lambda}\exp\left\{-\,\frac u2\left(t+\frac 1t\right)\right\}, 
\label{A_1} \\
&&\!\!\!\!\!\!\!\!\!\!\!\!\!\!\!\!\!\!\!\!
\psi_{l\,0}(kr) = \left(\frac{\pi kr}{2}\right)^{1/2}\!\!\!J_{l+{\frac 12}}(kr)
= \frac{1}{2i}\left[i^{-l}\chi_l(0-ikr)-i^l\chi_l(0+ikr)\right], 
\label{24_1} \\
&&\!\!\!\!\!\!\!\!\!\!\!\!\!\!\!\!\!\!\!\!
\mbox{with: }\; 
\int\limits^\infty_0\! dr\, \psi_{j\,0}(kr)\psi_{j\,0}(qr)=\frac{\pi}2\delta(q-k).  
\label{24_0} 
\end{eqnarray}
By making use of (\ref{24}) and (\ref{A_5}) for integer $l$ and $z=0-ikr$ one finds: 
\begin{equation}
\int\limits^\infty_0\!d\xi\,e^{z(1-\xi)}P_l(1-\xi)=\frac{\chi_l(-z)}z. 
\label{A_6}
\end{equation}
The following well-known expressions for Wronskians \cite{T,Ng,AR} $\forall\,j,\,l$
are used:
\begin{eqnarray}
&&\!\!\!\!\!\!\!\!\!\!\!\!\!\!\!\!\!\!\!\!
(\chi_j(ikR)\overset{\leftrightarrow}{\partial}_R \chi_j(-ikR) )=2ik, 
\,\;\mbox{ or: }\,\;
(\chi_j(z)\overset{\leftrightarrow}{\partial}_z\chi_j(-z))=2, 
\label{A_9} \\
&&\!\!\!\!\!\!\!\!\!\!\!\!\!\!\!\!\!\!\!\!
\frac{(\chi_l(ikR)\overset{\leftrightarrow}{\partial}_R \chi_j(-ikR))}{2ik}= 
1-\frac{\Delta_{jl}}{2ik}\int\limits^\infty_R\!\frac{dr}{r^2}\chi_l(ikr)\chi_j(-ikr)).
\label{A_10}
\end{eqnarray}
For the integral (\ref{A_10}) with integer $j,\,l$ the Eq. (\ref{24}) gives:   
\begin{eqnarray}
&&\!\!\!\!\!\!\!\!\!\!\!\!\!\!\!\!
1-\frac{\Delta_{jl}}{2ik}\int\limits^\infty_R\!\frac{dr}{r^2}\chi_l(ikr)\chi_j(-ikr))=
1+\Delta_{jl}\sum^{l+j}_{n=0}\frac{A_n(l,j)}{(n+1)(-2ikR)^{n+1}},
\label{A_7} \\
&&\!\!\!\!\!\!\!\!\!\!\!\!\!\!\!\!
A_n(l,j)=\sum^{\min(n,l)}_{s\,=\,\max(0,n-j)}\frac{(-1)^s}{s!(n-s)!}\frac{(l+s)!}{(l-s)!}
\frac{(j+n-s)!}{(j-n+s)!}.
\label{A_8} 
\end{eqnarray}


\section*{Acknowledgments}
Authors thank V. Naumov, D. Naumov, A. Rastegin and N. Ilyin for useful discussions 
and Reviewer for constructive comments and valuable advice.

\end{document}